\documentclass[preprint,12pt,authoryear]{elsarticle}
\usepackage[margin=2.5cm]{geometry}

\usepackage{url}
\usepackage{amsmath}
\usepackage{amsfonts}
\usepackage{amssymb}
\usepackage{amsthm}
\usepackage{bbm}
\usepackage{graphicx}
\usepackage{algorithm}
\usepackage{algpseudocode}
\usepackage{subcaption}
\usepackage{diagbox}
\usepackage{multirow, makecell}
\usepackage{bm}

\def\bigtimes{\mathop{\mathchoice{
   \vcenter{\hbox to10bp{\vrule height15bp width0pt \pdfliteral{
   q 1 J .8 w 0 1 m 10 14 l S 0 14 m 10 1 l S Q
}\hss}}}{
   \vcenter{\hbox to10bp{\kern1bp\vrule height10bp width0pt \pdfliteral{
   q 1 J .65 w 0 0 m 8 10 l S 0 10 m 8 0 l S Q
}\hss}}}{\times}{\times}
}}

\newtheorem{theorem}{Theorem}[section]
\newtheorem{lemma}[theorem]{Lemma}
\newtheorem{corollary}[theorem]{Corollary}

\journal{Computational Statistics \& Data Analysis}

\begin{document}

\begin{frontmatter}



\title{Thompson, Ulam, or Gauss? Multi-criteria recommendations for posterior probability computation methods in Bayesian response-adaptive trials} 


\author{Daniel Kaddaj\corref{cor1}}
\ead{dk620@cam.ac.uk}
\author{Stef Baas, Edwin Y.N. Tang, David S. Robertson, Lukas Pin\fnref{sen}}
\author{and Sofía S. Villar\fnref{sen}} 
\cortext[cor1]{Corresponding author}
\affiliation{organization={MRC Biostatistics Unit, University of Cambridge},
            addressline={East Forvie Building, Forvie Site, Robinson Way}, 
            city={Cambridge},
            postcode={CB2 0SR}, 
            country={United Kingdom}}
\fntext[sen]{Joint senior authors}

\begin{abstract}
Bayesian adaptive designs enable flexible clinical trials by adapting features based on accumulating data. Among these, Bayesian Response-Adaptive Randomization (BRAR) skews patient allocation towards more promising treatments based on interim data. Implementing BRAR requires the relatively quick evaluation of posterior probabilities. However, the limitations of existing closed-form solutions mean trials often rely on computationally intensive approximations which can impact accuracy and the scope of scenarios explored. While faster Gaussian approximations exist, their reliability is not guaranteed. Critically, the approximation method used is often poorly reported, and the literature lacks practical guidance for selecting and comparing these methods, particularly regarding the trade-offs between computational speed, inferential accuracy, and their implications for patient benefit.

In this paper, we focus on BRAR trials with binary endpoints, developing a novel algorithm that efficiently and exactly computes these posterior probabilities, enabling a robust assessment of existing approximation methods in use. Leveraging these exact computations, we establish a comprehensive benchmark for evaluating approximation methods based on their computational speed, patient benefit, and inferential accuracy. Our comprehensive analysis, conducted through a range of simulations in the two-armed case and a re-analysis of the three-armed Established Status Epilepticus Treatment Trial, reveals that the exact calculation algorithm is often the fastest, even for up to 12 treatment arms. Furthermore, we demonstrate that commonly used approximation methods can lead to significant power loss and Type I error rate inflation. We conclude by providing practical guidance to aid practitioners in selecting the most appropriate computation method for various clinical trial settings.

\end{abstract}






\end{frontmatter}



\section{Introduction}
\label{s:intro}
Adaptive designs aim to make clinical trials more patient-oriented and efficient by allowing pre-planned alterations to a study in response to accumulating data, while maintaining statistical validity and integrity \citep{burnett}. These designs have steadily gained popularity in the literature and are increasingly being used in recent practice \citep{lee}. Bayesian adaptive designs trigger adaptations to the experimental design based on the computation of the posterior probability of superiority~(PPS) for one of the treatments. In fact, one of the earliest described adaptive designs, Bayesian response-adaptive randomisation (BRAR), was first proposed by \citet{thompson} and involves allocating patients sequentially to one of two treatments proportional to their PPS. In 1933 it was not computationally feasible to approximate these probabilities, and computing them exactly with closed formulae required strict assumptions.

At present, Bayesian designs are a widely used type of adaptive design, ranging from dose-finding to confirmatory trials. For example, in dose-finding the continual reassessment method (CRM) \citep{oquigley1990} uses the posterior probability of a dose being the one closest to a target to assign the next cohort of patients to a dose. Recently, \citet*
{Pin2024backfilling} proposed using BRAR to allocate patients to doses in a patient-orientated way below the maximum tolerated dose to identify potential efficacy plateau in novel oncology treatments. Bayesian posterior probabilities may be calculated for multiple endpoints. Similarly, some designs include early stopping rules that are based on posterior probabilities of efficacy or futility \citep{gsponer}.
These designs involve the use of a complex posterior probability which, due to the inherent complexity and limitations of existing closed-form solutions for exact posterior probability computations, is usually approximated typically via computationally intensive simulations (e.g., \citet{berry}). 

While many adaptive designs, particularly Bayesian Response-Adaptive Randomization (BRAR), frequently utilize binary endpoints \citep{Pin2024software}, the potential for exact posterior probability computations in such settings is often overlooked. To our knowledge, precise closed-form formulae for the PPS of one treatment, even when using common Beta priors with integer-valued parameters, exist but are limited to the case of a maximum of four treatment arms and become computationally cumbersome for more than two treatments arms. Consequently, despite their potential benefits, particularly in small sample size settings where the accuracy of approximations is most critical, these exact methods remain less known and largely unutilized in practical applications.
Moreover, there is no (unique) software package that is typically used for such Bayesian response-adaptive trials, and no guidance on which approximation methods are most appropriate in different situations. Indeed, some trials do not straightforwardly provide the method of computation used. For example, \citet{tobacco} and \citet{ESET} don't give the methods used but make the codes available by email request, while \cite{endTB} detail the appointment  of an independent statistician employed for the calculations, and \citet{AGILE} give no information. 


In this paper, we address the computational challenges in BRAR trials, a setting known for frequent and impactful PPS calculations (see \citet{Pin2024software}). Our contribution lies in enabling exact computation of these probabilities and, crucially, in providing a comprehensive framework to understand what fundamentally differentiates these exact approaches from approximations, considering their implications for implementation time and inferential quality. While focused on BRAR, our findings extend relevance to other Bayesian adaptive designs.

Our contributions in this paper are as follows. First, we propose a novel algorithm to compute the PPS exactly for any number of treatment arms, with a computational time that is linear in the trial size. We then use this algorithm as a way to benchmark different approximation methods in practice, which allows us to derive clearer guidance for the use of these methods for practitioners to rely on when designing and implementing BRAR trials. More specifically, we focus on three methods, which to our knowledge are most common: the first is a simple but widely used (for example in \citet{ADORE}, \citet{REMAP-CAP} or \citet{CSPN}) approach based on Monte-Carlo simulation. The second is numerical integration, used, for example, in \citet{ARREST}. The third is a less common method using Gaussian approximation, but which promises improved computational performance, and was used in the Bayesian re-design in \citet{ga_trial}.

We show that as well as its accuracy benefits, our exact calculation algorithm can bring significant computational advantages, particularly for trials that frequently update their allocation probabilities such as early-phase or dose-finding trials. Indeed, while numerical integration is shown to be particularly computationally expensive, for two treatment arms exact calculation can be 12 times faster than Gaussian approximation and even faster for more arms, and 2290 times faster than the Monte-Carlo method. This is particularly useful since it allows for simulations to be carried out with more iterations or a greater range of parameters when designing a BRAR trial, the lack of which can in practice lead to uncertain power estimates (e.g. power estimates in \citet*{conor} based on only 1,000 simulations) or a failure to fully control the type I error rate as outlined in \citet*{baas}. Moreover, in terms of accuracy, we illustrate how Gaussian approximation can cause both significant inflation in the type I error rate and power loss, while the Monte-Carlo method even with 10,000 iterations can cause some power loss, therefore highlighting the importance of appropriately choosing a method of computation.

The paper is organized as follows. Section 2 introduces the model specification, assumptions and metrics used for comparing methods. This section also outlines the algorithm for the exact computation of posterior probabilities individually and over the course of a BRAR trial. Section 3 presents simulation studies for a 2-arm trial, introducing the derivation of the metrics suggested. Section 4 considers the case-study of a real-world 3-arm trial and our corresponding recommendations. Section 5 concludes the paper with a discussion, including practical guidance for the general multi-arm case.

\section{Methods}
\label{s:methods}

\subsection{Introductory notation}

In this section we introduce notation adapted from \citet{RAR_summary} and \citet{baas}. We consider clinical trials with a maximum trial size of $n$ patients labelled from 1 to $n$, and $k$ treatment arms labelled from 0 to $k-1$. The patients are sequentially assigned to exactly one treatment (i.e. if $i\geq i'$ then patient $i$ is allocated to a treatment no earlier than patient $i'$), with $A_{i,j}$ denoting a (possibly random) indicator which is 1 iff patient $i$ is assigned to treatment $j$. The response $Y_i\in\{0,1\}$ of patient $i$ to their treatment is observed before the next patient is allocated. The random number of patients up to and including patient $i$ that have been assigned to treatment $j$ is $N_{i,j}$, and the number of these patients who had a positive response to their treatment is $S_{i,j}$. We denote the trial history up to and including patient $i$ as ${\bm H}_i:=(\bm{A}_1,Y_1,\ldots,\bm{A}_i,Y_i)$. Moreover, for response probabilities $\bm{p}\in [0,1]^k$, let $Y_i\sim\text{Bernoulli}(p_{a_i})$ independent of ${\bm H}_{i-1}$, where $a_i$ is the treatment patient $i$ was assigned.

A \textit{response-adaptive} (RA) procedure uses knowledge of the trial's history to, for example, improve patient benefit by defining a function $\pi:\bigcup\limits_{i=0}^n(\{0,1\}^k\times\{0,1\})^{i}\rightarrow [0,1]^k$ such that $\pi_j({\bm H}_i):=\mathbb{P}(A_{i+1,j}=1\mid{\bm H}_i)$, i.e. one assigning randomisation probabilities based on trial history up to patient $i$. The joint probability distribution on the allocations and outcomes following from~$\pi$ is denoted~$\mathbb{P}^\pi,$ and the corresponding expectation and variance operators are denoted by~$\mathbb{E}^\pi$ and $\mathbb{V}^\pi,$ respectively.

A \textit{simple Bayesian response-adaptive randomisation} (S-BRAR) procedure is an RA procedure, as in \citet{thompson}, that uses the intuitive approach of assigning a patient to a particular treatment arm with the probability that given the known information at that time, that treatment is the best one. Therefore, we let $\mathbb{Q}:=\mathbb{Q}_0\times\ldots\times\mathbb{Q}_{k-1}$ be a prior on the response probabilities $\bm{p}$, and $\mathbb{Q}^i$ be the distribution of $\bm{p}$ given the trial history ${\bm H}_i$. In fact, $\mathbb{Q}^i$ decomposes nicely, and only depends on ${\bm H}_i$ through $N_{i,j}$ and $S_{i,j}$, which simplifies calculations (see Lemma \ref{lem:prior_decomp}). Then, 
\begin{equation}\label{eq:brar}
    \pi^{\mathbb{Q}}_{\text{S-BRAR},j}({\bm H}_i):=\mathbb{Q}^i(p_j>\max_{j'\neq j}p_{j'}).
\end{equation}
Response-adaptive procedures are useful when the aim of the trial is identify a superior option among many and there is considerable uncertainty about which of them could be the best one. The null hypothesis therefore corresponds to there being no treatment with a maximal probability of positive response, while the alternative hypotheses correspond to the different possible superior treatments, so that
\begin{equation*}
    H_0:\: \left| \max\{p_0,\ldots,p_{k-1}\}\right| >1; \: H_{A_j}:\: p_j>\max_{j'\neq j}p_{j'}.
\end{equation*}

We consider trials with early stopping rules and therefore a random sample size~(less than or equal to the maximum trial size~$n$). To this end, we define a \textit{rejection rule} $\mathcal{R}=\bigcup_{i=1}^n\bigcup_{j=0}^{k-1}\mathcal{R}^i_j,$ where $\mathcal{R}^i_j\subseteq\{0,1\}^{(k+1)i}$ is the set of trial histories where the trial is stopped after patient $i$ in favour of hypothesis $H_{A_j}$.

Since they are most often used in such BRAR trials \citep{edwin}, we utilize the posterior probabilities as test statistics for the various alternative hypotheses. More precisely, for a treatment $j\in \{0,\ldots,k-1\}$, we define
$T_j({\bm H}_i):=\mathbb{Q}^i(p_{j}>\max_{j'\neq j}p_{j'})$.
We can then choose some threshold probability $c\in[0,1]$ (see Appendix~\ref{appc} for details) to determine a rejection rule $\mathcal{R}(\mathcal{P},c)$ which rejects in favour of hypothesis $H_{A_j}$ if $\max_{i\in\mathcal{P}}T_j({\bm H}_i)>c$, where $\mathcal{P}\subseteq \{1,\dots, n\}$ is the non-empty set of patients after whose allocation an analysis occurs to determine whether to reject the null hypothesis, so that early stopping is permitted. Such a rejection rule was used in practice in several trials such as \citet{ESET}, as well as \citet{ADORE}.

Note that the test statistics after patient $i$ are the same as the randomisation probabilities for patient $i+1$, so that we can simply consider computing randomisation probabilities for a trial whose number of patients is increased by one, which will then also give us the test statistics for the original trial.

In this paper, we restrict ourselves to independent Beta priors with integer-valued parameters on the response probabilities ${\bm p}$, which generalise uniform priors and allow for easier implementation and analysis by making use of conjugacy. Let $\overset{d}{\equiv}$ denote the equality of two distributions. Indeed, if for all $0\leq j<k$, $\mathbb{Q}_j\overset{d}{\equiv}\text{Beta}(a_j,b_j)$ for some $\bm{a},\bm{b}\in[1,\infty)^k$, it can be shown (see Lemma \ref{lem:conj}) that $\mathbb{Q}^i\overset{d}{\equiv}\bigtimes_{j=0}^{k-1}\text{Beta}(a_j+S_{i,j},b_j+N_{i,j}-S_{i,j})$. In terms of the trial, this tells us that the posterior distribution on $p_j$ at any point in the trial is a Beta distribution depending only on its prior, and the number of positive and negative responses to treatment $j$ until that point, independent of other treatments. This holds regardless of the RA procedure chosen. To simplify notation, assuming that $a_j,b_j$ are integer valued for all $j$, the posterior probabilities we wish to calculate can then be written in the form $P_{k}(\bm{x},\bm{y}):=\mathbb{P}(X>\max_{j}X_{j})$ where $X\sim\text{Beta}(x_0,x_1)$ and $X_j\sim\text{Beta}(y_{2j},y_{2j+1})$ are independent, $\bm{x}\in\mathbb{N}^{2}$, and $\bm{y}\in\mathbb{N}^{2(k-1)}$.

\subsection{Computing randomisation probabilities}

Computing randomization probabilities for the above-described setting is not straightforward, often necessitating approximations. In this paper, we consider both exact and various approximated methods available for this context, which we describe next.

\subsubsection{Exact evaluation}

Though not often used, an exact formula for $P_{k}(\bm{x},\bm{y})$ in the case where $k=2$ was derived by both \citet{liebermeister} and \citet{thompson}. Recently \citet{miller_blog} suggested a sightly more computationally tractable form:
\begin{equation}\label{eq:miller}
   P_{2}((x_2,x_3),(x_0,x_1)) = \sum_{i=0}^{x_2-1} \frac{B(x_0+i,x_1+x_3)}{(x_3+i)B(i,x_3)B(x_0,x_1)},
\end{equation}
where $B$ is the Beta function. Using symmetry, the sum can be rewritten to have only $\min(x_0,x_1,x_2,x_3)$ terms.

Here, we propose a recursive approach allows us to calculate $P_{k}(\bm{x},\bm{y})$ exactly for any $k$ in a computationally efficient way. In fact, a more general result allows for a particularly fast sequential calculation of all the posterior probabilities required to run an S-BRAR trial with Beta priors having positive integer parameters:
\begin{theorem}\label{thm}
    For $\bm x\in\mathbb{N}^{2k}$, $S\subset\{0,1,\dots,k-1\}$, and $\bm{x}_{S'}$ the vector $\bm{x}$ with entries $x_{2s}$ and $x_{2s+1}$ removed for all $s\in S$, let 
    \begin{equation*}
        P(\bm{x};S):=P_{k-|S|+1}((\textstyle \sum_{s\in S}x_{2s},\sum_{s\in S}x_{2s+1}),\bm{x}_{S'}).
    \end{equation*} Let $P(\bm{x};\{0,\dots,k-1\}):=1$. Then given a sequence $\{\bm{x}^0,\dots,\bm{x}^n\}$ in $\mathbb{N}^{2k}$ increasing by exactly 1 in exactly one entry, Algorithm~\ref{algo} calculates $\{P(\bm{x}^{n'};\{j\}):j\in\{0,1,\dots,k-1\}, n'\in\{0,1,\dots,n\}\}$ in $\mathcal{O}\left(k2^k\left(n-2k+\sum_{i=0}^{2k-1}x^0_i\right)\right)$ operations (including the evaluation of a Beta function). In particular, this is the set of posterior probabilities required to run an S-BRAR  trial with $n$ patients, Beta priors with parameters corresponding to $\bm{x}^0$, and path of states equal to $\bm x^0, \bm x^1, .., \bm x^n$.
\end{theorem}
The proof of Theorem~\ref{thm} can be found in Appendix \ref{proof_thm}.
This computation also gives all the required posterior probabilities if the randomisation is done in batches of more than one, or if an initial period of equal randomisation is used since after this period we simply update the priors and then apply the theorem.

Now, if we seek to find $P_{k}(\bm{x},\bm{y})$ in isolation, we can consider an S-BRAR trial with uniform priors, and a path of states starting at a vector of ones corresponding to the uniform priors and ending at the state $(\bm x,\bm y)$. Then, $P_{k}(\bm{x},\bm{y})$ is a posterior probability at the end of the trial, and so applying Theorem~\ref{thm} we get:
\begin{corollary}
    $P_{k}(\bm{x},\bm{y})$ can be calculated exactly in $\mathcal{O}\left(k2^k\left(x_0+x_1+\sum_{i=0}^{2k-3}y_i-2k\right)\right)$ operations for any $k\geq2$, $\bm{x}\in\mathbb{N}^{2}$ and $\bm{y}\in\mathbb{N}^{2(k-1)}$.
\end{corollary}
Note that when $k=2$, this gives an operational complexity of $\mathcal{O}(x_0+x_1+y_0+y_1)$, whereas the adapted form of Equation \eqref{eq:miller} gives one of $\mathcal{O}(\min(x_0,x_1,y_0,y_1))$.

The recursive structure which underpins these results relies on the following lemma, which extends an approach from \citet{cook_rec}. 
\begin{lemma}\label{lem:rec}
    If $0\leq j,j'\leq k-1$, $k\geq2$, $s\in\{0,1\}$, $\bm{x}\in\mathbb{N}^{2k}$, then
    \begin{equation}\label{eq:rand}
        P(\bm{x}+\bm{e}_{j,s};\{j'\}))-P(\bm{x};\{j'\}) =\begin{cases}
            (-1)^{s+1}\sum_{j''\neq j}\frac{b_{j,j''}(\bm{x})}{x_{2j+s}}P(\bm{x};\{j,j''\}) &\text{if $j=j'$,}\\ (-1)^s\frac{b_{j,j'}(\bm{x})}{x_{2j+s}}P(\bm{x};\{j,j'\}) &\text{if $j\neq j'$},
        \end{cases}
\end{equation}
where $e_{l,j,s}:=\mathbbm{1}(l=2j+s)$ for $l=0,\ldots,2k-1$ represents an increment of one, and $b_{j,j'}(\bm{x}):=\frac{B(x_{2j}+x_{2j'},x_{2j+1}+x_{2j'+1})}{B(x_{2j},x_{2j+1})B(x_{2j'},x_{2j'+1})}$.
\end{lemma}
The proof of Lemma~\ref{lem:rec} can be found in Appendix \ref{proof_lemma_diff}.
This seemingly technical result can be more intuitively understood by analogy to an S-BRAR trial. Suppose one has randomised the latest patient according to $\{P(\bm{x};\{j'\}):j'\in\{0,\dots,k-1\}\}$ which is known, they are assigned to treatment $j$, and exhibit a response $s$. Then the next patient will then be randomised according to $\{P(\bm{x}+\bm{e}_{j,s};\{j'\}):j'\in\{0,\dots,k-1\}\}$. As such, Lemma \ref{lem:rec} allows us to reduce the calculation of subsequent posterior probabilities to calculation of the set $\{P(\bm{x};\{j,j'\}):j'\neq j\}$, which are posterior probabilities for a $k-1$ armed trial. Though we could iterate until the two-armed case which can be evaluated explicitly, Theorem \ref{thm} instead proposes storing $P(\bm{x};S)$ for all non-empty $S\subseteq\{0,\dots,k-1\}$. Then, by noticing that Equation \eqref{eq:rand} can be extended to any non-empty $S$, not just $S=\{j'\}$, we can calculate the updated $\{P(\bm{x}+\bm{e}_{j,s};S):S\subseteq\{0,\dots,k-1\},S\neq\emptyset\}$ in only $\mathcal{O}(k2^k)$ operations since we stored all the posterior probabilities necessary for the update. 

\begin{algorithm}
\caption{Algorithm to exactly compute all posterior probabilities in a BRAR trial with integer Beta priors}\label{algo}
\begin{algorithmic}[1]
\State \textbf{Input:} A sequence $\bm{x}^0,\dots,\bm{x}^n$ in $\mathbb{N}^{2k}$ increasing by exactly 1 in exactly one entry
\State $n' \gets \sum_{i=0}^{2k-1}x_i^0-2k$
\For{$j\gets 0$ to $k-1$}
\State $P^{(-n')}(\{j\})\gets \frac{1}{k}$
\EndFor
\For{$l\gets 2$ to $k-1$}
\For{$S\subseteq\{0,\dots,k-1\}$ s.t. $|S|=l$}
\State $P^{(-n')}(S)\gets \frac{1}{k-l+1}+(k-l)\sum_{j=1}^{l-1}\left(\frac{B(k-l+j,j+2)}{jB(j,j+1)}-\frac{B(k-l+j,j+1)}{jB(j,j)}\right)$
\EndFor
\EndFor
\State Create a dummy sequence $\bm{1}_{2k}\rightarrow\bm{x}^{-n'},\dots,\bm{x}^0$ increasing by exactly 1 in exactly one entry
\For{$\tilde{n}\gets -n'$ to $n-1$}
\State Let $j\in\{0,\dots,k-1\},s\in\{0,1\}$ be s.t. $x^{{\tilde{n}}+1}_{2j+s}-x^{\tilde{n}}_{2j+s}=1$
\For{$l\gets 1$ to $k-2$}
\For{$S\subseteq\{0,\dots,k-1\}$ s.t. $|S|=l$}
\If{$j\in S$}
\State $P^{(\tilde{n}+1)}(S) \gets P^{(\tilde{n})}(S)+(-1)^{s+1}\sum_{j'\notin S}\frac{b_{j'}(\bm{x}^{\tilde{n}};S)P^{(\tilde{n})}(S\cup\{j'\})}{\sum_{j''\in S}x^{\tilde{n}}_{2j''+s}}$
\State where $b_{j'}(\bm{x}^{\tilde{n}};S):=\frac{B(x^{\tilde{n}}_{2j'}+\sum_{i\in S}x^{\tilde{n}}_{2i},x^{\tilde{n}}_{2j'+1}+\sum_{i\in S}x^{\tilde{n}}_{2i+1})}{B(x^{\tilde{n}}_{2j'},x^{\tilde{n}}_{2j'+1})B(\sum_{i\in S}x^{\tilde{n}}_{2i},\sum_{i\in S}x^{\tilde{n}}_{2i+1})}$
\Else
\State $P^{(\tilde{n}+1)}(S) \gets P^{(\tilde{n})}(S)+(-1)^{s}\frac{b_{j}(\bm{x}^{\tilde{n}};S)P^{(\tilde{n})}(S\cup\{j\})}{x^{\tilde{n}}_{2j+s}}$
\EndIf
\EndFor
\EndFor
\State \textbf{Store and output:} $P(\bm{x}^{n+1};\{j\})=P^{(n+1)}(\{j\})$
\State Clear $P^{(n)}$
\EndFor
\end{algorithmic}
\end{algorithm}

\subsubsection{Approximation methods}

An alternative, which is commonly used, is to approximate $P(\bm{x};\{j\})$ with one of the following methods:
\begin{itemize}
    \item \textit{Numerical integration} (NI), used for example in \citet{ARREST}, can be used with the required integral made explicit in Lemma \ref{lem:prior_decomp}.
    \item \textit{Repeated sampling} (RS) is an intuitive approach first developed by \citet{ulam}, used for example in \citet{ADORE}, which involves taking $K$ independent samples $\bm{X}_1,\ldots,\bm{X}_K$ from $\mathbb{P}:=\bigtimes_{j'=0}^{k-1}\text{Beta}(x_{2j'},x_{2j'+1})$ and observing the fraction of these in which $j$ was the best treatment arm, so that
    \begin{equation}
        \hat{P}^{\text{RS}}(\bm{x};\{j\}):=\frac{1}{K}\sum_{K'=1}^K\mathbbm{1}(X_{j,K'}>\max_{j'\neq j}X_{j',K'}).
    \end{equation}
    \item \textit{Gaussian approximation} (GA) was proposed for $k=2$ in \citet{cook_ga} and used for the Bayesian re-design in \citet{ga_trial}, which approximates the Beta distributions in the posterior probability by normal distributions with the same mean and variance, giving for the general case of $k\geq 2$
\begin{equation}
    P^{\text{GA}}(\bm{x};\{j\}) = F^j(0,\ldots,0)
\end{equation}
where $F^j$ is the cumulative distribution function for the multivariate ($k-1$ dimensional) normal distribution $\mathcal{N}(\boldsymbol{\mu}^{(j)},\boldsymbol{\Sigma}^{(j)})$ with $\boldsymbol{\mu}^{(j)} = \boldsymbol{\mu}_{-j} - m_j\bm 1_{k-1}, 
\boldsymbol{\Sigma}^{(j)} = \boldsymbol{D}_{-j}+ s^2_j \bm 1_{k-1} \bm 1_{k-1}^\top$
where $\boldsymbol{\mu}_{-j}$ is the vector $[m_1,\dots,m_k]$ with entry $m_j$ removed, $\bm 1_{k-1}$ is the all-ones vector of length $k-1$, $\boldsymbol{D}_{-j}$ is the diagonal matrix with diagonal $[s^2_0,\dots s^2_k]$ with $s_j^2$ removed, and $m_{j'} = \frac{x_{2j'}}{x_{2j'}+x_{2j'+1}}$, $s_{j'}^2 = \frac{x_{2j'}x_{2j'+1}}{(x_{2j'}+x_{2j'+1})^2(x_{2j'}+x_{2j'+1}+1)}$ for all $j'$. In fact, a different form of Gaussian approximation for $k=2$ was described earlier by \citet{altham}, but it does not easily generalise to the case when $k>2$ and, to the best of the authors' knowledge, has not been used in a BRAR trial in practice.
\end{itemize}

\subsection{Why should we care about posterior probability computation?}

 First, we assess the fundamental trade-off between the accuracy and computational complexity of approximating a single probability, $P(\bm{x};\{j\})$, in isolation. Then, and based on the previously described assessment, we benchmark approximation methods from two distinct perspectives. In the broader context of a trial, these posterior probabilities serve two primary roles: one when used as test statistics, in which case  we quantify how their approximation affects critical inferential properties such as the Type I error rate and power and second, when these probabilities are employed to guide the sequential patient randomization throughout the trial, in which case we investigate the potential compounding of approximation errors and their subsequent impact on allocations and on other key operating characteristics.

Since the posterior probabilities are used in different ways for testing and randomisation, and this can in practice be done on different systems, different methods of calculation might be appropriate for each. In terms of testing, to implement the Bayesian-based rejection rule introduced, the posterior probability needs to be evaluated explicitly. For randomisation, however, computing the value of $\pi^{\mathbb{Q}}_{\text{S-BRAR}}({\bm H}_i)$ is not strictly necessary, as this can be achieved instead by sampling from $\mathbb{Q}^i$ and assigning patient $i+1$ to the arm corresponding to the maximal sample \citep{russo}.

However, the randomisation probabilities need to be known explicitly for various regularised versions of S-BRAR. Tuning procedures have been proposed to mitigate the effects of extreme and highly variable probabilities by scaling the randomisation probabilities. A tuned BRAR (T-BRAR) procedure is therefore defined as an S-BRAR procedure which has had a \textit{tuning procedure} ${\bm t}:[0,1]\rightarrow\mathbb{R}_+^k$ applied to it such that 
\begin{equation}
    \pi^{\mathbb{Q}}_{\text{T-BRAR},j}({\bm H}_i):=\frac{t_j(\pi^{\mathbb{Q}}_{\text{S-BRAR},j}({\bm H}_i))}{\sum_{j'=0}^{k-1}t_{j'}(\pi^{\mathbb{Q}}_{\text{S-BRAR},j'}({\bm H}_i))}.
\end{equation}

One such approach, used among others by \citet{ESET}, aims to favour not only arms whose treatment effect is likely better, but also those whose variance is greater, since not accounting for this is a significant cause of extreme allocation proportions. For a parameter $m\in\mathbb{N}$, we define such an approach as \textit{variance-scaling} with 
\begin{equation}\label{eq:vs}
    t_{j,\text{VS}(m)}(\pi^{\mathbb{Q}}_{\text{S-BRAR},j}({\bm H}_i)):= \left(\frac{\pi^{\mathbb{Q}}_{\text{S-BRAR},j}({\bm H}_i)\text{Var}_{\mathbb{Q}^i}(p_j)}{N_{i,j}+1}\right)^{\frac{1}{m}}.
\end{equation}

\subsection{Testing}

In designing a trial with rejection rule $\mathcal{R}(\mathcal{P},c)$ we typically choose $c$ such that the type I error rate is below a nominal significance level $\alpha$. If our null hypothesis can be reduced to $p_0=\ldots=p_{k-1}$, either if $k=2$ as is the case Section \ref{s:sim} or by adding further alternative hypotheses as is the case in Section \ref{s:cs}, an \textit{unconditional exact} test UX$(\mathcal{P},\alpha)$ does this by setting $c$ to $c(\mathcal{P},\alpha)$, which is defined as the smallest $c$ to control the type I error rate below $\alpha$.

Finding this control parameter is not straightforward and it constitutes another useful application of exact computation. To see this we first look at a slightly simpler approach which seeks instead to control the type I error rate below $\alpha$ at only one response probability $p$ within the null hypothesis. The \textit{pointwise posterior} test PP($\mathcal{P},p,\alpha$) sets $c$ to $c(\mathcal{P},p,\alpha)$, which is instead defined as the smallest $c$ to control the type I error rate below $\alpha$ if $p_0=\ldots=p_{k-1}=p$. By noting that $c(\mathcal{P},\alpha)=\sup_{p\in[0,1]}c(\mathcal{P},p,\alpha)$, we can therefore see how finding $c(\mathcal{P},p,\alpha)$ helps us find $c(\mathcal{P},\alpha)$, and, although not used in this paper, Algorithm 2 from \citet{baas} outlines how to do this efficiently.

Algorithm 1 from the same paper outlines a way to find $c(\mathcal{P},p,\alpha)$ which involves calculating all possible posterior probabilities $P(\bm{x};\{j\})$ for all $0\leq j<k$ and $\bm{x}\in\mathbb{N}^{2k}$ such that for all $j'$, $x_{2j'}\geq a_{j'}$, $x_{2j'+1}\geq b_{j'}$ and $\sum_{l=0}^{2k-1}x_l< n+\sum_{j'=0}^{k-1}a_{j'}+b_{j'}$. If we use uniform priors, this corresponds to $\mathcal{O}(n^{2k-1})$ probabilities, making it clear that for larger $n$ or $k$ efficient computation of these probabilities can be important. 

Moreover, we can see that not only the implementation of the allocation procedure and the calculation of the test statistic, but also the derivation of the exact test itself depends on accurate computation of these posterior probabilities, which could impact on inference. In the case of a deterministic approximation such as GA, if this is used both for calculating the test statistics and $c(\mathcal{P},p,\alpha)$, then the type I error rate will remain controlled below $\alpha$ at the required point, though the remainder of the type~I error rate profile and power may be different.

\subsection{Operating characteristics}\label{subsec:oc}

We consider a range of operating characteristics (OCs). Letting $\bar{T}$ be the patient after which the trial is stopped, following \citet{baas_notation} we define:
\begin{itemize}
    \item The \textit{type I error rate} is the probability that we reject the null hypothesis given it is true, i.e., $\mathbb{P}^{\pi}(\max_{j\in\{1,\dots, k\}}T_j({\bm H}_{\bar{T}})>c)$.
    \item The \textit{power} is the probability we reject the null hypothesis in favour of $H_{A_j}$, i.e., $$\mathbb{P}^{\pi}(T_j({\bm H}_{\bar{T}})>c),$$  when $j$ is the superior treatment arm.
    \item $\mathbb{E}^{\pi}\left[N_{j,\bar{T}}+(n-\bar{T})\mathbbm{1}(T_j({\bm H}_{\bar{T}})>c)\right]$ is the number of \textit{expected patients allocated to the superior arm} (EPASA), where $j$ is the superior treatment arm. The assignment of patients both before and after possible optional stopping is considered so that, for example, stopping early in favour of the best treatment is not penalised.
    \item $\mathbb{V}^{\pi}\left(N_{j,\bar{T}}+(n-\bar{T})\mathbbm{1}(T_j({\bm H}_{\bar{T}})>c)\right)$ is the \textit{variance in patients allocated to the superior arm} (VPASA), where $j$ is the superior treatment arm. This aims to capture the variance in the benefit to patients within the trial, since the EPASA may be high but fail to show that in a number of cases there can be harm to patients.
\end{itemize}
Since BRAR trials have been shown to produce biased treatment effect estimates \citep*{thall}, and so are better suited to differentiate between test hypotheses of interest rather than estimating treatment effects, we have not included OCs based on such treatment effect estimates.  

Note that all these OCs can be expressed as the expected value of a function of ${\bm H}_{\bar{T}}$, with the VPASA done through the formula for variance in terms of expectations. If we further specify that the rejection rule is of the form $\mathcal{R}(\mathcal{P},c)$, all these OCs can then be calculated exactly as outlined in \citet{baas}. The same can be done to find the OCs if GA is used either to randomise patients during the trial or to approximate the test statistics since GA is deterministic. However, if the approximation method used for the posterior probabilities is random, as for RS, or if the the number of patients or treatment arms is too large for such an approach to be computationally feasible, the same cannot be done. 

In such cases, we can estimate the OCs by simulation. In particular, we can run some large number $K$ of trials with the desired approximation method and use the average of the function of ${\bm H}_{\bar{T}}$ within the operating characteristic's definition over all these trials as an estimate for the target OC. Given a particular trial set-up, for the type I error rate and power each sample trial will end with a 0 or 1, so that these can be thought of as i.i.d. Bernoulli random variables $B_1,\ldots,B_K$ with some parameter $q\in[0,1]$ which is the OC we wish to find. Then a Chernoff bound \citep{chernoff} can be used to get an upper bound on the probability that the absolute distance between our estimate and the true value is greater than some $\epsilon$ we can choose. For values of $q$ that are not very small, the following from \citet{kearns_saul} gives us a tighter bound:
\begin{equation}\label{eq:ks}
    \mathbb{P}\left(\left|\frac{1}{K}\sum_{K'=1}^KB_{K'}-q\right|>\varepsilon\right)\leq2\left(\frac{1-q}{q}\right)^{-\frac{K\cdot\varepsilon^2}{1-2q}}
\end{equation}
for all $\varepsilon>0$. This is particularly useful since RS was used in practice to estimate test statistics in \citet{ADORE}, though with $10^6$ samples instead of the $10^4$ studied here.

\section{Investigation of two-armed S-BRAR}
\label{s:sim}

We will first consider a two-armed trial with a binary endpoint in greater detail. We use an S-BRAR procedure with uniform priors on the response probabilities, and no early stopping is incorporated.

\subsection{Assessing the effect of approximating a single posterior probability}

Suppose we wish to calculate a single posterior probability in isolation. For the sake of simplicity, denote $P_{2}((c+1,d+1),(a+1,b+1))$ by $P(a,b,c,d)$. In this case, the principal trade-off is between accuracy and computational feasibility. We compare $P_{NI}$ (obtained using the \texttt{dblquad} function from Python's \texttt{scipy} package \citep{dblquad}), $\hat{P}_{RS}$, $P_{GA}$, and Equation \eqref{eq:miller} for $P$ as computational methods. For repeated sampling, we use $K=10^4$ simulations, which comes at a significant run-time cost despite remaining accuracy concerns.

We compute $P(a,b,c,d)$ using a 14-inch MacBook Pro, M3, 18 GB RAM, macOS Sonoma 18.1, Python in VSCode. Note that for constant $n:=a+b+c+d$, which corresponds to the number of patients in the trial up to that point, the exact calculation is slowest if $a=b=c=d$, while both GA and RS have approximately constant run-times over the different values of $a,b,c,d$. Figure~\ref{fig:iso}(a) therefore compares the methods in the worst case for exact calculation. For $n\leq 1000$, both GA and exact calculation are at least 60 times faster than NI, and RS is at least as fast as NI for $n\gtrapprox 50$, with all these differences growing as $n$ increases. RS has an approximately constant run-time with respect to $n$, and for $n\leq 1000$ has GA being about 190 times faster, and exact calculation being at least 12 times faster. GA also has an approximately constant run-time with respect to $n$, and is slower than exact calculation only for $n\lessapprox 50$. The run-time for exact calculation is approximately linear in $n$, and is about 15 times slower than GA for $n=1000$. Moreover, since concerns around computational feasibility likely arise when calculating a sequence or a range of such probabilities, it is possible to pre-load, for example, all the values of the beta-function with positive integer arguments of at most 1000. Having done this, exact calculation becomes faster than GA for $n\lessapprox 400$, and at $n=1000$ exact calculation is about twice as slow. 

As for accuracy, the NI used provides an accuracy estimate of at most $10^{-7}$ for $n\leq 1000$. As the mean of i.i.d. Bernoulli random variables, $10^4\hat{P}_{RS}\sim \text{Bin}(10^4,P)$. We can therefore use the formula from \citet{blyth} to find that
\begin{equation}
    \mathbb{E}\left[|\hat{P}_{RS}-P|\right]=2\binom{10^4-1}{l-1}P^l(1-P)^{10^4-l+1}
\end{equation}
where $l$ is such that $l-1\leq P\times 10^4 < l$. Importantly, this is a function of $P$, and takes the maximal value of $3.99\times10^{-3}$ when $P=0.5$. As a result, we find empirically that for any fixed number of patients on each arm (i.e. fixed $a+b$ and $c+d$) such that $n\leq 200$, the largest value of the mean absolute error $\mathbb{E}\left[|\hat{P}_{RS}(a,b,c,d)-P(a,b,c,d)|\right]$ is approximately constant at $3.99\times10^{-3}$, and this would be expected to hold for larger $n$. In fact, this gives us a bound on the mean absolute error from repeated sampling for an arbitrary $K$ which applies for any value of $j$, $k$, and $\bm{x}$, and can be used to select $K$ as required:
\begin{equation}\label{eq:rs}
    \mathbb{E}\left[|\hat{P}^{RS}_{j,k}(\bm{x})-P_{j,k}(\bm{x})|\right]\leq\binom{K-1}{\lfloor\frac{K}{2}\rfloor}2^{-K}.
\end{equation}

\begin{figure}
    \centering
    \begin{subfigure}[b]{0.47\linewidth}
        \includegraphics[width=\linewidth]{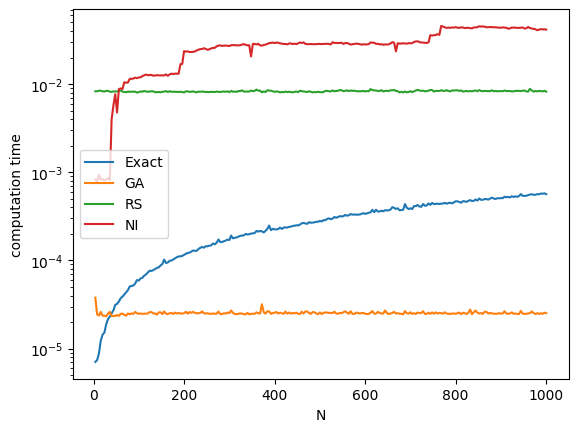}
        \caption{}
    \end{subfigure}
    \begin{subfigure}[b]{0.47\linewidth}
        \includegraphics[width=\linewidth]{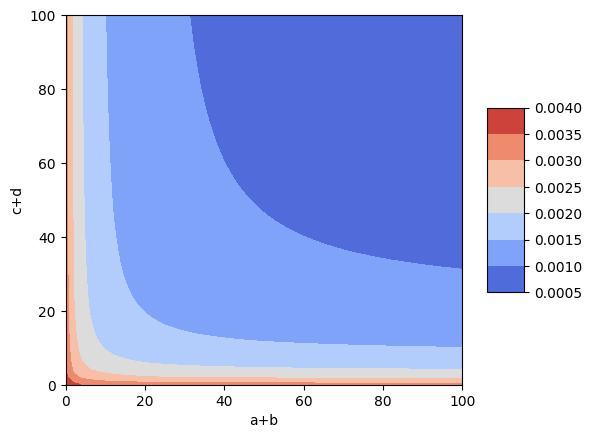}
        \caption{}
    \end{subfigure}
    \begin{subfigure}[b]{0.47\linewidth}
        \includegraphics[width=\linewidth]{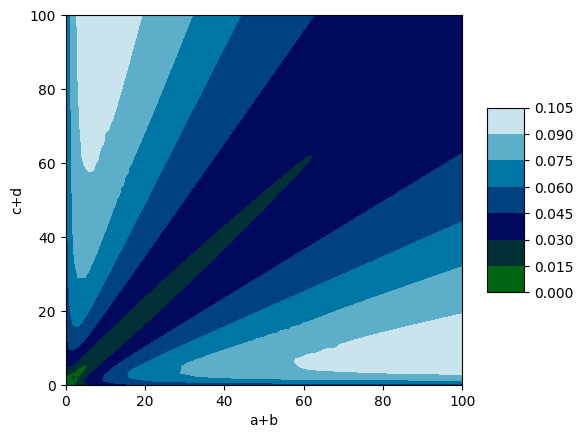}
        \caption{}
    \end{subfigure}
    \begin{subfigure}[b]{0.47\linewidth}
        \includegraphics[width=\linewidth]{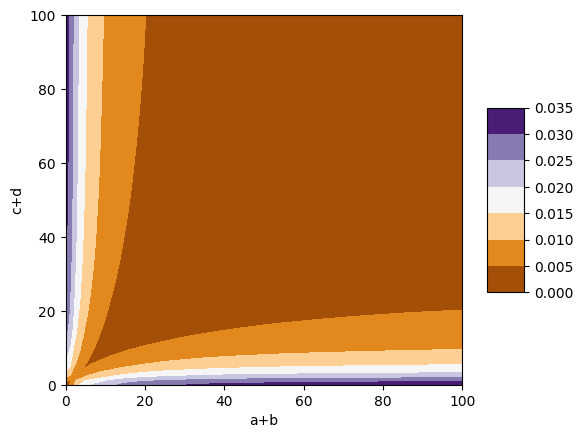}
        \caption{}
    \end{subfigure}
    \caption{(a) The logarithm of the time taken (in seconds) to compute or approximate $P(a,b,c,d)$ using different methods if $a=b=c=d$ as $n$ changes. For different fixed values of $a+b$ and $c+d$, representing the number of patients on each treatment arm: (b) the mean value (over all configurations) of $\mathbb{E}\left[|\hat{P}_{RS}(a,b,c,d)-P(a,b,c,d)|\right]$; (c) the largest value of $|P_{GA}(a,b,c,d)-P(a,b,c,d)|$; (d) the mean value of $|P_{GA}(a,b,c,d)-P(a,b,c,d)|$.}
    \label{fig:iso}
\end{figure}

Figures~\ref{fig:iso}(b)--(d) show the mean and worst-case absolute error given a fixed number of patients per arm induced by using GA to approximate $P$, as well as the mean absolute error induced by using RS. The latter was largest if the number of patients assigned to one arm was small, exceeding 0.003 as long as no patients were assigned to one of the arms. This is of particular concern for an RA procedure such as S-BRAR which can result in extreme allocations.

Moreover, regarding GA we can see that for $n\geq 113$, the absolute error can exceed 0.1, such that in general GA provides very weak accuracy guarantees. Further, a larger number of patients induces a greater error, which might not be expected given the results from \citet{cook_ga} and the intuition that a normal approximation of a beta distribution is more accurate if the parameters are larger. This is because while the number of patients on each arm is fixed at a possibly larger number, the number of either successes or failures (i.e. individual parameters) can be smaller provided the other is larger. Nevertheless, in the worst case, the treatment allocation ratio is highly improbable under S-BRAR given the number of successes and failures on each arm (for example $a,b,c,d = 85,8,0,7$ for $n=100$). From the mean-case, however, we still see smaller but significant errors arising, as for RS, at more extreme treatment allocation ratios. Those for GA are also still almost an order of magnitude larger than those for RS with $K=10^4$.

\subsection{Assessing the effect of approximating posterior probabilities in the full trial}

\subsubsection{Computational considerations}

We wish to distinguish the relative contributions of approximating posterior probabilities either to determine randomisation probabilities, or to determine test statistics. In practice, as we saw, one would approximate either both or neither, but the method of approximation for each need not be the same.

Indeed, while computing the set of all randomisation probabilities in a trial with $n$ patients can be done the same time as calculating $P(n,n,n,n)$ independently, if using GA or RS calculating $P$ at each update has a fixed cost, so that for $n=1000$ exact computation is now about 12 or 2290 times faster respectively, with this multiple even larger for smaller values of $n$. Since there is only one possible stopping point of the trial, however, calculating the test statistic at the end of the trial would only require one instance of this fixed cost, so that using GA for randomisation and RS for testing would be 23 times slower than exact calculation, while the other way around would be 2280 times slower. 

In particular, this suggests that though GA is much faster than RS, if RS helps with accuracy the cost of its use in testing may be smaller, while the cost of its use in randomisation is rather larger. Yet this advantage would diminish if the number of possible early stopping points were increased.

\subsubsection{Approximating for implementation of the randomisation procedure}

First, we report the impact on the OCs concerning within-trial patient benefit. The change in EPASA if GA was used to approximate posterior probabilities for randomisation over the trial compared to exact probabilities ranges from about -8.00 to 1.49 (i.e. from a decrease by 4.17\% to an increase by 0.757\%). The difference in both absolute and percentage terms is largest if one of $p_0,p_1$ is near 0.03 and the other is near 0.4, it is positive over most of the parameter space so long as neither response probability is too large, but reaches its minimum in a smaller area when one of $p_0,p_1$ is near 1 and the other is near 0.96. The corresponding change in VPASA ranges from about -878 to 242, but is positive on almost all the parameter space and only negative if one of $p_0,p_1$ is near 0 and the other is near 1. In percentage terms, however, the largest decreases in VPASA were by as much as 26.1\% and occurred if either $p_0$ or $p_1$ were near 1, and the largest increases were by as much as 55.0\% and occurred if one of $p_0,p_1$ was near 0.6 and the other near 1. To gain more intuition, the range for the corresponding change in standard deviation in the proportion of patients assigned to the superior treatment is about -0.0407 to 0.0107. Overall, we see that using GA for randomisation instead of exact calculation has a somewhat limited impact on patient benefit and that generally it leads to a slight increase in the mean number of patients within the trial benefiting but with a rather larger variance. As for using RS to approximate randomisation probabilities, we saw in the previous section that OCs could only be estimated using simulation. For EPASA and VPASA, we cannot directly use a Chernoff bound to guarantee the accuracy of such estimates as we do not know the distribution of the simulated outcomes and bounds on them are large, so that Hoeffding's inequality \citep{hoeffding} gives very weak guarantees. Nevertheless, we expect an even smaller impact on these metrics because of the greater accuracy of RS.

\begin{figure}
    \centering
    \begin{subfigure}[b]{0.47\linewidth}
        \includegraphics[width=\linewidth]{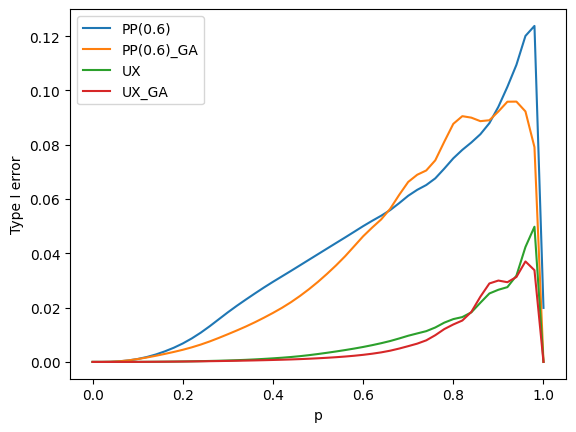}
        \caption{}
    \end{subfigure}
    \begin{subfigure}[b]{0.47\linewidth}
        \includegraphics[width=\linewidth]{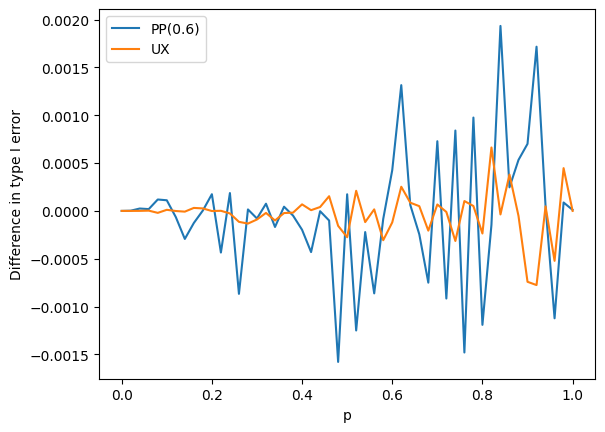}
        \caption{}
    \end{subfigure}
    \begin{subfigure}[b]{0.47\linewidth}
        \includegraphics[width=\linewidth]{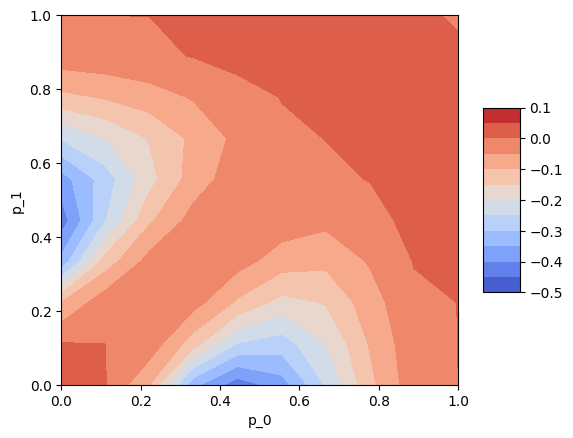}
        \caption{}
    \end{subfigure}
    \begin{subfigure}[b]{0.47\linewidth}
        \includegraphics[width=\linewidth]{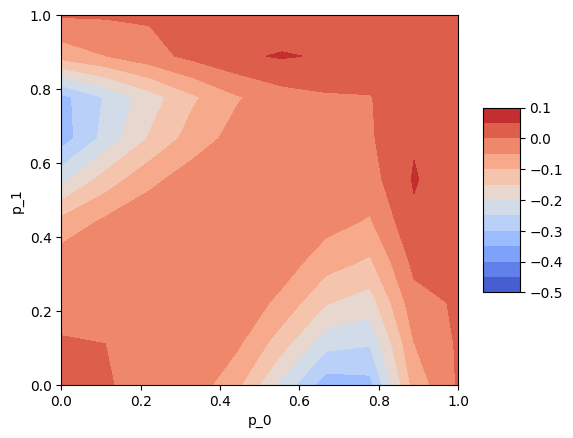}
        \caption{}
    \end{subfigure}
    \begin{subfigure}[b]{0.47\linewidth}
        \includegraphics[width=\linewidth]{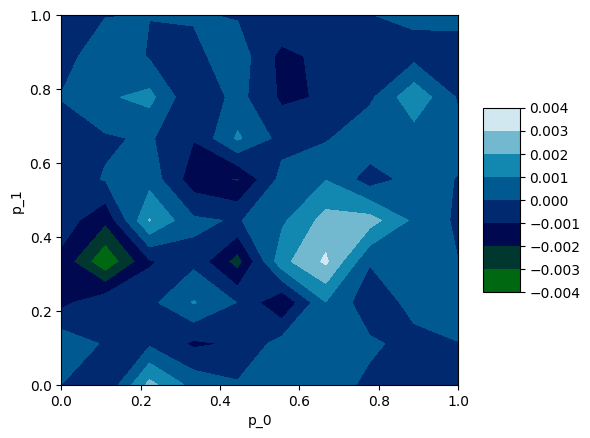}
        \caption{}
    \end{subfigure}
    \begin{subfigure}[b]{0.47\linewidth}
        \includegraphics[width=\linewidth]{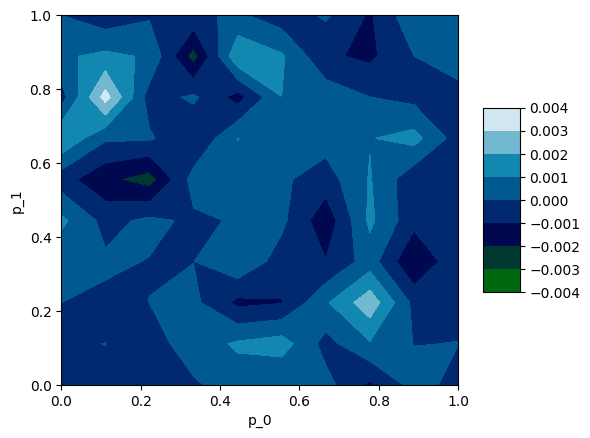}
        \caption{}
    \end{subfigure}
    \caption{For $n=200$ and different values of $p_0$ and $p_1$: (a) the type I error rate if $p_0=p_1=p$ for tests PP(0.6), UX using exact probabilities and GA for randomisation; (b) the estimated difference in type I error rate if RS or exact probabilities are used for randomisation; the power change from using GA for randomisation for the (c) PP(0.6), (d) UX, tests; and the estimated power change from using RS for randomisation for the (e) PP(0.6), (f) UX, tests.}\label{fig:rand}
\end{figure}

In terms of the impact of approximating randomisation probabilities on inference, Figure~\ref{fig:rand} shows how the type I error rate and power change for $N=200$ and different tests. For convenience we will omit $\mathcal{P}=\{n\}$ and $\alpha = 0.05$ for the notation for tests. In particular, Figure~\ref{fig:rand}(a) plots the type I error rate from PP(0.6) and UX, with and without GA for randomisation. We can see that GA can cause both type I error rate inflation and deflation, depending on the test and response probability $p$, and that for the same $p$ a different test can yield either inflation or deflation. The largest inflation for the PP test is 0.0127 and that for UX is 0.00369. In the case of UX, however, under GA the type I error rate remains below 0.05, and its maximal value decreases, so that in this sense GA can be said not to negatively affect the type I error rate. 

In Figure~\ref{fig:rand}(b) for the same tests the type I error rate for RS used for randomisation was simulated based on $10^5$ instances and the exact type I error rate was subtracted from this. Using Equation \eqref{eq:ks} we find that for $10^5$ samples, our estimate has a 95\% chance of being within 0.0043 of the true value, so there is not statistically significant evidence that using RS for randomisation causes type I error rate inflation, and in fact there is evidence that if such inflation exists it will not exceed 0.0063.

On the other hand, Figures~\ref{fig:rand}(c)--(d) plot the change in power from using GA for randomisation for the same PP(0.6) and UX tests respectively. As such, we see that for both there is power loss over most of the parameter space, and it can be substantial for some $p_0,p_1$, exceeding as much as 0.42 for PP(0.6). Overall, then, PP(0.6) and UX exhibit little type I error rate inflation, but at the cost of possible substantial power reduction with only a much smaller possible improvement in power depending on the response probabilities.

With the same estimation method as in Figure~\ref{fig:rand}(b), the estimated change in power from using RS for randomisation is plotted in Figures~\ref{fig:rand}(e)--(f). Since all the values are smaller in absolute value than 0.004, again there is no significant evidence of power reduction and evidence that any improvement in power does not exceed 0.0083. We can therefore conclude that while GA can have significant negative distorting effects on inference, this is not the case with RS. However, using RS for randomisation is much more computationally costly, so that there is a necessary trade-off in the approximation method chosen.

\subsubsection{Assessing the impact of approximating for statistical testing}

\begin{figure}
    \centering
    \begin{subfigure}[b]{0.47\linewidth}
        \includegraphics[width=\linewidth]{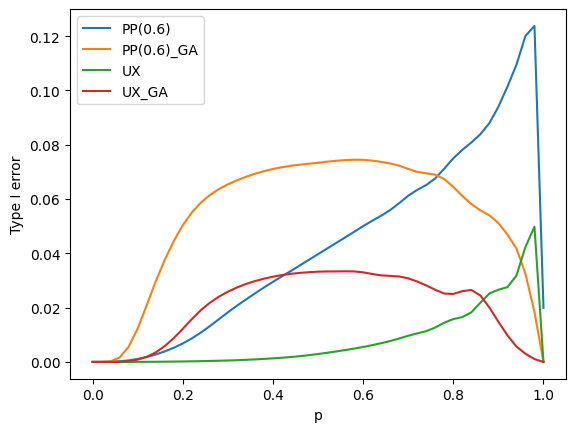}
        \caption{}
    \end{subfigure}
    \begin{subfigure}[b]{0.47\linewidth}
        \includegraphics[width=\linewidth]{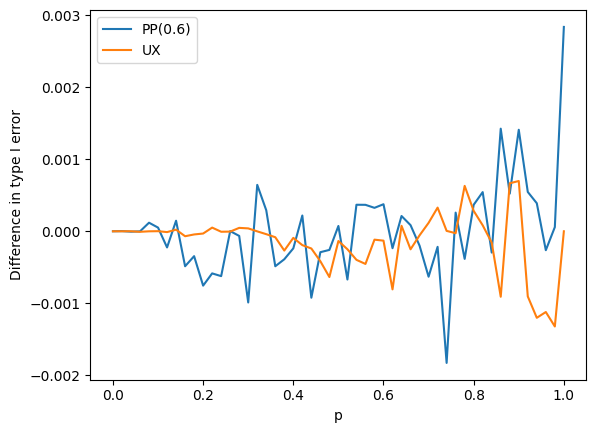}
        \caption{}
    \end{subfigure}
    \begin{subfigure}[b]{0.47\linewidth}
        \includegraphics[width=\linewidth]{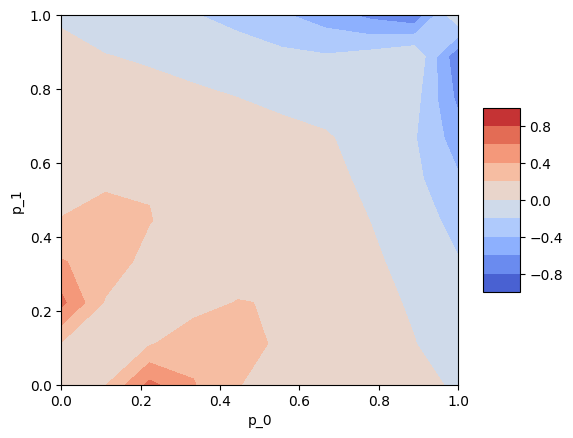}
        \caption{}
    \end{subfigure}
    \begin{subfigure}[b]{0.47\linewidth}
        \includegraphics[width=\linewidth]{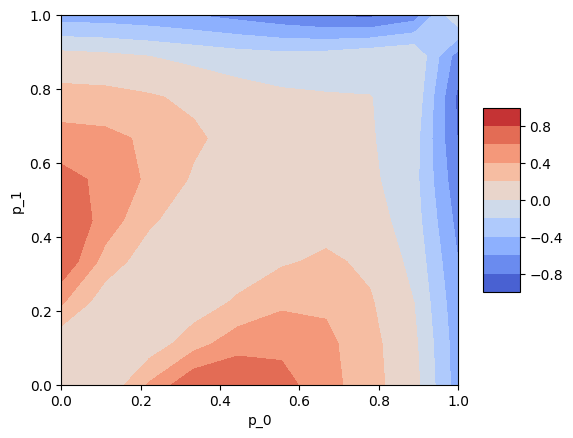}
        \caption{}
    \end{subfigure}
    \begin{subfigure}[b]{0.47\linewidth}
        \includegraphics[width=\linewidth]{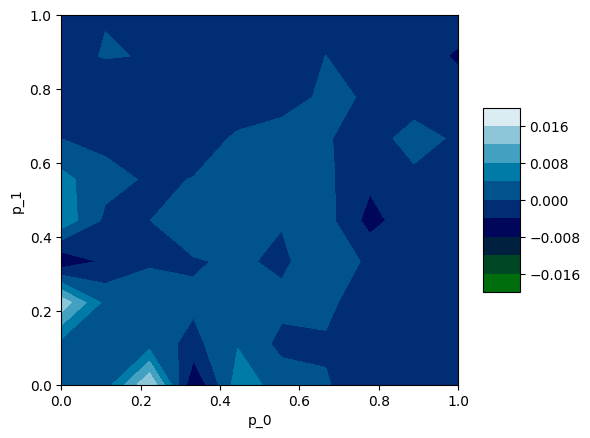}
        \caption{}
    \end{subfigure}
    \begin{subfigure}[b]{0.47\linewidth}
        \includegraphics[width=\linewidth]{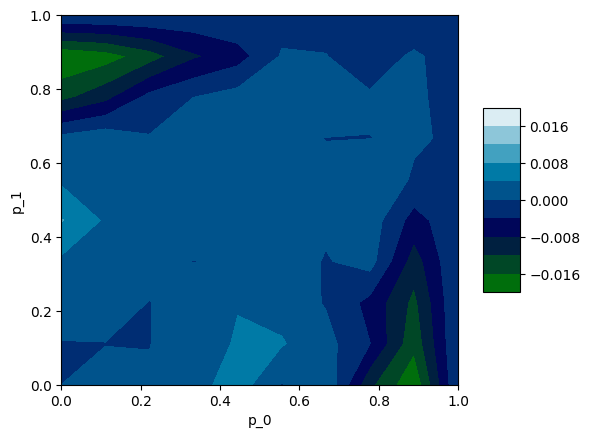}
        \caption{}
    \end{subfigure}
    \caption{For $n=200$ and different values of $p_0$ and $p_1$: (a) the type I error rate if $p_0=p_1=p$ for tests PP(0.6), UX using exact probabilities and GA for testing; (b) the estimated difference in the type I error rate if RS or exact probabilities are used for testing; the power change from using GA for testing for the (c) PP(0.6), (d) UX, tests; and the estimated power change from using RS for testing for the (e) PP(0.6), (f) UX, tests.}\label{fig:test}
\end{figure}

If the test statistic $T_{1}=1-T_{0}$ is approximated, we can look at the impact made by both RS and GA on the type I error rate and power for both PP(0.6) and UX, and $n=200$, as is done in Figure~\ref{fig:test}. Indeed, Figure~\ref{fig:test}(a) shows the graph of the type I error rate with and without GA for these tests, and we can see that unlike in Figure~\ref{fig:rand}(a) the impact is both more significant and alters the qualitative behaviour of the graph. For both tests, large $p$ show substantial error deflation while for other $p$ not near 0 there is large error inflation, reaching 0.048 for PP(0.6) and 0.031 for UX. As for RS, Figure~\ref{fig:test}(b) plots the same quantity as Figure~\ref{fig:rand}(b) but with RS used for testing instead of randomisation. For both tests, the difference in the type I error rate is less than 0.003 so there is no statistically significant evidence of error inflation.

Further, Figures~\ref{fig:test}(c)--(d) show the change in power from using GA for inference for these same two tests respectively. In both cases the power can change dramatically, increasing over most of the space, especially in (d), but decreasing if either $p_0$ or $p_1$ is large. Under the PP(0.6) test, the power increase is largest if one of $p_0,p_1$ is close to 0 while the other is close to 0.2, and reaches 0.65, while the decrease is largest if one of $p_0,p_1$ is close to 1 while the other is close to 0.9, and reaches 0.75. For the UX test these values are even larger, with the largest power increase of 0.81 when one of $p_0,p_1$ is close to 0 while the other is close to 0.4, and the largest power decrease of 0.84 when one of $p_0,p_1$ is close to 1 while the other is close to 0.8. It is also noteworthy that though this correspondence is rather inexact, the response probabilities that improve power when GA is used for randomisation are also typically those where power is also reduced when GA is used for testing.

Finally, Figures~\ref{fig:test}(e)--(f) show the estimated changes in power from using RS for testing. For both tests, larger response probabilities typically induce power loss while smaller ones typically induce power gain. Nevertheless, we can see that the behaviour becomes qualitatively different depending on the test used, with PP(0.6) having a maximal power gain of almost 0.02 and a maximal power loss of only at most 0.008, and the reverse for UX. Even subtracting 0.0043 from each of these to account for the estimates' 95\% confidence intervals, we can see that RS can lead to non-trivial power increases and decreases for the PP(0.6) and UX tests respectively, though still much smaller than the corresponding effect of using GA.

As a result, if we wish to keep power similar to the original trial design, using RS to calculate test statistics is much better than GA, and the additional computational cost of this is less significant. Yet some bias is still introduced as compared with the use of RS for randomisation, so that the lower computational penalty comes at the cost of accuracy. Meanwhile, though GA is computationally cheaper, it also induces significant distortions especially if used for test statistics, likely because a small error in the test statistic can directly lead to a different decision at the end of the trial in a way that is not possible if randomisation probabilities are slightly changed. In the case of UX however, the type I error rate remains controlled at the required level over the whole parameter space and there are substantial power gains for smaller response probabilities, so that for this test GA may actually be advisable if $p_0$ and $p_1$ are unlikely to be large. As such, even if the same Bayesian-based test statistic is used, it is important to also consider the objectives of the test, as this can make a difference to the suitability of different approximation methods.

\section{Case Study: The ESET trial}
\label{s:cs}

Real trials are likely to be more complex than the set-up from the previous section, so we consider an example with a larger number of patients, more than two treatment arms, randomisation in blocks, an initial period of equal randomisation called the \textit{burn-in}, possible early stopping, and tuning. We aim to deduce practical recommendations for the appropriate choice of calculation method of posterior probabilities in a realistic situation.

The \textit{Established Status Epilepticus Treatment} (ESET) trial \citep{ESET} compares the efficacy of available treatments for benzodiazepine-refractory status epilepticus by indentifying the best and/or worst treatment among fosphenytoin (fPHT), levetiracetam (LVT), and valproic acid (VPA) in patients older than 2 years. Since all these treatments are in use there is no control, and for convenience we label them 0, 1, and 2 respectively. Up to 795 patients are recruited, but to account for re-enrolment, treatment crossover, and missing data, 720 unique patients were assumed to have data for final analysis. After randomising the first 300 patients equally to each treatment arm, T-BRAR with the variance-scaling tuning procedure and parameter $m=2$ in Equation \eqref{eq:vs} is used but with randomisation probabilities only updated every 100 patients. 

The trial considered the null hypothesis for which there is neither a unique best nor unique worst of three treatments, and alternative hypotheses representing either one treatment being best, one being worst, or both, so that 
\begin{align*}
    H_0: p_0=p_1=p_2; \: & H_{\overline{j}}: p_j>p_{j'} \text{ for }j'\neq j; \\ H_{\underline{j}}: p_j<p_{j'} \text{ for }j'\neq j; \: & H_{\overline{j},\underline{j}'}: p_j>p_{j''}>p_j \text{ for }j''\neq j,j'.
\end{align*}
Moreover, we define a test statistic to correspond to a particular treatment being worst,
$T_j'({\bm H}_i):=\mathbb{Q}^i(p_{j}<\min_{j'\neq j}p_{j'})$.
This can be calculated in the same way as $T_j$ since $\mathbb{Q}^i(p_{j}<\min_{j'\neq j}p_{j'})=\mathbb{Q}^i(1-p_{j}>\max_{j'\neq j}1-p_{j'})$. If no early stopping has occurred, at the end of the trial the null hypothesis is rejected in favour of $H_{\overline{j},\underline{j}'}$ if $T_j({\bm H}_n)>0.975$ and $T_{j'}'({\bm H}_n)>0.975$ where $j$ and $j'$ are distinct arms that have not been dropped representing the best and worst treatment respectively. Otherwise, the null hypothesis is rejected in favour of $H_{\overline{j}}$ if $T_j({\bm H}_n)>0.975$, or of $H_{\underline{j}}$ if  $T_j'({\bm H}_n)>0.975$, where $j$ is defined as before.

Interim analyses are carried out after each block of 100 patients is randomised according to the T-BRAR procedure, so that this is first done after 400 patients have been allocated treatment. If the posterior probability of treatment $j$ being optimal at that point, i.e. $T_j({\bm H}_i)$, is greater than 0.975, the null hypothesis is rejected in favour of $H_{\overline{j}}$. As such, the design prioritises stopping the trial early having identified the superior treatment over identifying the inferior treatment. Further, if at an interim analysis $\mathbb{Q}^i(p_j<0.25)\geq 0.95$, then treatment arm $j$ is dropped. If all treatments are dropped, the trial stops for futility. It should be noted that the original design has a further futility stopping rule \citep[outlined on p. 91 of][]{ESET} that has been omitted here for simplicity.

We are interested in the effect of changing the frequency of interim adaptation/analysis, which we call $b$, and the number of patients equally assigned to each arm at the start of the trial, which we call $B$. The configuration previously described corresponds to $B=b=100$. In this trial, prior to randomising the next block of patients, it is determined whether or not the trial can be stopped early by computing the test statistics. As such, there is no computational advantage to calculating the randomisation probabilities faster but less accurately, since other than at the first block of patients these probabilities are just test statistics that would already have been calculated using a slower and more accurate method.

\begin{table}
    \caption{Maximal time (in seconds) taken to compute 10,000 ESET trials with different posterior probability calculation methods, burn-in period lengths $B$, and frequencies of interim analyses $b$.}
    \label{tab:comp_time}
    \centering
    \begin{tabular}{cl r r r r} 
    \cline{3-6}
    \noalign{\vskip\doublerulesep\vskip-\arrayrulewidth}
    \noalign{\vskip\doublerulesep\vskip-\arrayrulewidth}
    \cline{3-6}
    \multicolumn{2}{r}{} & $b$=100 & 20 & 5 & 1 \\ 
    \hline
    \hline
    \multirowcell{3}{Exact} & $B$=100 & 196 & 194 & 205 & 281 \\
    & 50 & 193 & 200 & 209 & 294 \\
    & 0 & 191 & 196 & 216 & 290 \\ 
    \hline
    \multirowcell{3}{GA} & 100 & 8.96 & 28.5 & 106 & 529 \\
    & 50 & 9.99 & 43.0 & 152 & 692 \\
    & 0 & 12.7 & 49.7 & 189 & 853 \\ 
    \hline
    \multirowcell{3}{RS} & 100 & 855 & 2740 & 10100 & 49300 \\
    & 50 & 967 & 3720 & 14000 & 64500 \\
    & 0 & 1210 & 4540 & 17100 & 76500 \\ 
    \hline
    \hline
    \end{tabular}
\end{table}

Table~\ref{tab:comp_time} compares how long it takes to compute $10^4$ simulations of an ESET trial (in the computationally-worst case it doesn't stop early for each simulation) for different values of $B$, $b$, and different methods if they are used both for randomisation and testing. We can see that generally for non-exact methods, as $B$ decreases the computational time increases, though by less than a factor of 2, while for exact calculation the change is less significant. As for decreasing $b$, again this leads to slower computation in general, though the effect is much smaller for exact calculation. For the other methods, however, the computational penalty of a more sequential trial (i.e. a smaller $b$) is particularly stark, with $b=1$ being about as much as 67 times slower than $b=100$ if either GA or RS is used. Further, using GA is faster than exact calculation for $b\geq 5$, but slower if $b=1$, with exact calculation performing better for smaller $B$. Using RS, on the other hand, is much slower than exact calculation for the range of $B$ and $b$ considered, with computation slowed down by about as much as 264 times. This makes it particularly difficult to use simulations to accurately calculate various operating characteristics, or to explore a range of different response probabilities. For example, for a fully sequential trial even with $B=100$, it takes about 137 hours to calculate type I error rate for a single set of response probabilities to the accuracy used in Table~\ref{tab:rej_rate}.

\begin{table}
    \caption{Type I error rate and power of ESET trials based on $10^5$ simulations with different randomisation and testing calculation methods, burn-in period lengths $B$, and frequencies of interim analyses $b$. The nominal significance level~$\alpha$ was set to~$5\%.$
    \label{tab:rej_rate}}
    \centering
    \begin{tabular}{cl r r r r r r r r} 
    \cline{3-10}
    \noalign{\vskip\doublerulesep\vskip-\arrayrulewidth}
    \noalign{\vskip\doublerulesep\vskip-\arrayrulewidth}
    \cline{3-10}
    \multicolumn{2}{r}{} & \multicolumn{4}{c}{\centering Exact} & \multicolumn{4}{c}{\centering GA} \\
    \cline{3-10}
    \multicolumn{2}{r}{} & $b$=100 & 20 & 5 & 1 & 100 & 20 & 5 & 1 \\ 
    \hline
    \hline
    \multirowcell{3}{Type I error rate\\ for $p_0,p_1,p_2$ \\ = 0.5,0.5,0.5 (\%)} & $B$=100 & 3.80 & 5.06 & 5.80 & 6.50 & 4.18 & 5.51 & 6.36 & 6.93 \\
    & 50 & 4.82 & 6.51 & 7.87 & 8.69 & 5.14 & 7.15 & 8.57 & 9.38 \\
    & 0 & 5.88 & 8.97 & 11.52 & 13.90 & 6.46 & 11.44 & 15.33 & 18.85 \\ 
    \hline
    \multirowcell{3}{Power for \\ $p_0,p_1,p_2$ \\ = 0.5,0.5,0.65 (\%)} & 100 & 90.73 & 92.34 & 93.12 & 93.64 & 90.83 & 92.28 & 93.33 & 93.68 \\
    & 50 & 90.97 & 92.91 & 93.61 & 94.01 & 91.01 & 92.78 & 93.51 & 94.03 \\
    & 0 & 91.53 & 92.90 & 93.41 & 93.53 & 91.51 & 92.77 & 93.07 & 92.67 \\ 
    \hline
    \multirowcell{3}{Power for \\ $p_0,p_1,p_2$ \\ = 0.5,0.65,0.65 (\%)} & 100 & 66.19 & 63.17 & 61.48 & 60.23 & 66.45 & 63.45 & 61.74 & 60.70 \\
    & 50 & 62.89 & 58.47 & 55.89 & 54.63 & 62.95 & 58.94 & 56.49 & 54.92 \\
    & 0 & 60.72 & 54.88 & 51.26 & 48.36 & 60.86 & 54.96 & 50.50 & 47.23 \\ 
    \hline
    \hline
    \end{tabular}
\end{table}

Since we saw that calculating the type I error rate or power to a suitable accuracy is not computationally feasible for more sequential trials if RS is used for randomisation, we will consider the impact on accuracy of GA as compared with exact calculation. For inferential concerns, Table~\ref{tab:rej_rate} shows the type I error rate and power of an ESET trial for different values of $B$ and $b$, depending on whether GA or exact calculation is used. Each of its entries is obtained using $10^5$ simulations, so that the same confidence interval of radius 0.0043 as in the previous section applies, while noting that this is in fact a conservative bound if the probabilities are small as in the case of the type I error rate. For example, if the estimated probability is less than 0.05 or 0.1, the radius can be shrunk to 0.022 or 0.034 respectively. The type I error rate was considered for $p_0=p_1=p_2=0.5$ since this was the null scenario the original trial was designed under. We can establish that in this case, for the values considered other than $B,b=$100,100 or 100,20 or 50,100, using GA instead of exact calculation induces statistically significant type I error rate inflation. This error rate inflation is worst for the smallest $B$ or $b$, both in absolute and percentage terms, so that this cannot be attributed simply to the underlying type I error rate increasing as $B$ and $b$ are varied. It is especially bad if there is no burn-in period, in which case the type I error rate increases by 0.0247 (a 27.6\% increase) if $b=20$, and by 0.495 (a 35.6\% increase) if $b=1$.

Comparing power when using GA instead of exact calculation, if $p_0,p_1,p_2=0.5,0.5,0.65$ as in the original trial designed to have a power exceeding 0.9, there is no statistically significant improvement. On the other hand, for a fully sequential trial with $B=0$ using GA leads to a decrease in power of 0.0086. As such, the increased type I error rate still leads at worst to a reduction in power and at best to an insignificant improvement. This is also the case if we aim to identify the inferior treatment arm when $p_0,p_1,p_2=0.5,0.65,0.65$.

In terms of the patient benefit metrics EPASA and VPASA, as for power and type I error rate the large number of states, induced by the large trial size, prohibits exact calculation. Relying on simulations, however, prohibits finding a sharp bound for the estimation error since these OCs can take large values. Nevertheless, even in the case when using GA instead of exact calculation alters the trial most significantly, namely when $B=0$ and $b=1$, using GA instead of exact calculation for $p_0,p_1,p_2=0.5,0.5,$ $0.65$ only increased the EPASA by 0.933 and the VPASA by 3026 (increases by 0.16\% and 21\% respectively) corresponding to an increase in the standard deviation in the proportion of patients assigned to the superior treatment of 0.0167, on the basis of $10^5$ simulations. As such, as in the previous section the impact of approximating posterior probabilities on patient benefit is rather limited, but GA does improve the average outcome slightly by increasing its variance rather more. This effect is, as expected, less severe if either $b$ or $B$ are increased, for example with the EPASA increasing by 0.440 and the standard deviation in the proportion of patients assigned to the superior treatment increasing by $4.01\times10^{-5}$ if $B=b=100$.  

\section{Discussion: Toward general recommendations}
\label{s:discuss}

Overall, in this paper we have introduced a new algorithm for efficiently and exactly computing posterior probabilities in a BRAR trial with integer Beta priors. This allowed us to have a benchmark against which to compare different approximation methods for the calculation of these probabilities using a range of metrics, and so to quantify their impact on patient benefit and inference. In particular, we showed this impact to be significant enough to warrant careful thought about the choice of method prior to designing a trial, especially for more sequential or aggressively adaptive designs that use these probabilities.

It can be observed, from comparing the results in Sections \ref{s:sim} and \ref{s:cs}, that the accuracies of different posterior probability computation methods, especially relative to each other, are not substantially changed by different numbers of treatment arms. In particular, assuming uniform priors are used, the maximal number of patients in the trial is fixed, and interim analyses are used both to randomise patients in blocks and for early stopping:
\begin{itemize}
    \item Using GA can pose a severe risk of type I error rate inflation and power loss, both for more frequently adaptive trials, but especially if a short burn-in is used.
    \item Using RS with 10,000 iterations can cause some power loss for frequently adapting trials with a short burn-in, though only to a rather limited extent and for specific response probabilities.
\end{itemize}

Nevertheless, because of the range of metrics of interest (such as type I error rate inflation, power loss, and change in EPASA or VPASA) whose accuracy could be distorted by approximation, as well as the significant variation in these based on the success probabilities, it is not straightforward to provide a single objective function for accuracy which could then be balanced against computational time in an optimisation framework. Moreover, while practitioners could use the framework outlined here to compare the different computational methods in their specific set-up, this would involve having carried out all the computations exactly, so that having borne this (possible) computational cost, there is no need to use approximation methods as well.

In order to make heuristic recommendations for choosing a method of computation in a general multi-armed BRAR trial, therefore, we need to assess their relative computational speeds for any number of treatment arms, as opposed only two or three arms as in Sections \ref{s:sim} and \ref{s:cs}. Assuming uniform priors are used, the maximal number of patients in the trial is fixed, and interim analyses are used both to randomise patients in blocks and for early stopping, by Theorem \ref{thm} and by construction,
\begin{equation*}
    \text{Maximal computation time}=\begin{cases}
        nf_{EX}(k) &\text{for exact calculation}, \\
        \lceil\frac{n-kB}{b}\rceil f_{GA}(k) &\text{for GA}, \\
        \lceil\frac{n-kB}{b}\rceil Kc_{RS} &\text{for RS},
    \end{cases}
\end{equation*}
for some functions $f_{EX},f_{GA}$ and constant $c_{RS}$. While these may vary slightly depending on features such as the tuning procedure or an additional inferiority test as in Section \ref{s:cs}, using the same computer as in Section \ref{s:sim} and S-BRAR with $n=20$, $B=0$, $b=1$ and averaging over 100 iterations, we find that
\begin{align*}
    -\log f_{EX}(2,3,\dots,8,12,13) &\approx  (10.8,10.0,9.32,8.54,7.73,6.87,6.05,2.83,2.03), \\
    -\log f_{GA}(2,3,\dots,8) &\approx (8.93,8.60,5.92,4.74,3.83,2.71,1.94), \\
    c_{RS} &\approx 7\times 10^{-4}.
\end{align*}
These values can be used, or similarly found for other $k$, to estimate computation times and balance this against the aforementioned accuracy considerations. Note in particular that $f_{EX}\approx e^{0.803k-12.49}$. Also, if $k$ is increased with $kB$, the total number of patients assigned to the burn-in phase, kept constant, the maximal computation time for GA will increase more rapidly than for exact computation, so that in this respect exact calculation improves relative to GA as the number of treatment arms increases.

\subsection{Guidance}

Exact calculation of posterior probabilities in a BRAR trial has the advantage of not needing to be chosen like the number of simulations in RS, or to be verified as being safe like GA due to its significant potential distortions, but can sometimes come at a computational cost. To tune RS one can consult Equation \eqref{eq:rs} and set $K$ based on the error desired. The potential adverse inferential errors induced by approximation depend substantially on the success probabilities, and can lead to different recommendations based on the aim of the test. Nevertheless, a less accurate approximation has a more limited impact in randomisation probabilities than in test statistics, though both can be significant. As for patient benefit, approximation does not significantly impact this, though for more aggressively adaptive trials GA improves average patient prospects with increased variability. 

\begin{table}
    \caption{Heuristic recommendations for a method to compute posterior probabilities in a BRAR trial based on: the number of treatment arms; whether accuracy, computational speed, or a mixture are prioritised; the length of the burn-in period; and the frequency of interim analyses at which patients are randomised and early stopping is permitted.}
    \label{tab:summary}
    \centering
    \begin{tabular}{ccrrrr} 
    \cline{3-6}
    \noalign{\vskip\doublerulesep\vskip-\arrayrulewidth}
    \noalign{\vskip\doublerulesep\vskip-\arrayrulewidth}
    \cline{3-6}
    \multicolumn{2}{r}{} & \multicolumn{2}{c}{\centering Infrequent interim analyses} & \multicolumn{2}{c}{\centering Frequent interim analyses} \\
    \cline{3-6}
    \multicolumn{2}{r}{} & Longer burn-in & Shorter burn-in  & Longer burn-in & Shorter burn-in \\ 
    \hline
    \hline
    \multirowcell{3}{$\leq7$ arms} & Acc. & \textbf{Exact} & \textbf{Exact} & \textbf{Exact} & \textbf{Exact} \\
    \cline{3-6}
    & Mix & GA & \textbf{Exact} & \textbf{Exact} & \textbf{Exact} \\
    \cline{3-6}
    & Comp. & GA & \textbf{Exact}/GA & \textbf{Exact}/GA & \textbf{Exact} \\
    \hline
    \multirowcell{3}{8-12 arms} & Acc. & \textbf{Exact} & \textbf{Exact} & \textbf{Exact} & \textbf{Exact} \\
    \cline{3-6}
    & Mix & RS & \textbf{Exact} & \textbf{Exact} & \textbf{Exact} \\
    \cline{3-6}
    & Comp. & RS & \textbf{Exact}/RS & \textbf{Exact}/RS & \textbf{Exact} \\
    \hline
    \multirowcell{3}{$\geq13$ arms} & Acc. & \textbf{Exact} & \textbf{Exact} & \textbf{Exact} & \textbf{Exact} \\
    \cline{3-6}
    & Mix & RS & RS & RS & \textbf{Exact}/RS \\
    \cline{3-6}
    & Comp. & RS & RS & RS & RS \\
    \hline
    \hline
    \end{tabular}
\end{table}

If $K$ is set to $10^4$ for RS, uniform priors are used, the maximal number of patients in the trial is fixed, and interim analyses are used both to randomise patients in blocks and for early stopping, Table~\ref{tab:summary}, obtained by evaluating computational times for different $b,B$ and balancing these against the accuracy considerations outlined in this section, presents a heuristic to choose the computation method considering the number of treatment arms, the frequency of randomisation, the length of the burn-in period, as well as different possible priorities. ``Acc." corresponds to sacrificing computational efficiency for the sake of accuracy as might be the case in a confirmatory trial, while ``Comp." corresponds to prioritising computational speed over accuracy (both within reasonable bounds) as might be the case when wanting to explore a larger parameter space for an exploratory trial. ``Mix" corresponds to a desire to balance these considerations. 

A key takeaway is that in a large fraction of cases, and in particular for trials with frequent interim analyses, a small burn-in period, and at most 12 treatment arms, exact evaluation is likely most suitable regardless of the priority chosen. From the database of BRAR trials compiled by \citet{Pin2024software}, we can see that most of these are in fact multi-arm, so that the new exact computation method in this paper is indeed of some practical utility. Moreover, we see that there are BRAR trials with more than 7, and also more than 12 arms, so that these distinctions are not purely theoretical.

\section*{Acknowledgements}

\sloppy
This work was supported by the UK Medical Research Council [grant numbers MC\_UU\_00002/15, MC\_UU\_00040/03, MC\_UU\_00002/14], MRC Biostatistics Unit Core Studentship and the Cusanuswerk e.V. (to LP).  SSV discloses an advisory role with PhaseV. For the purpose of open access, the author has applied a Creative Commons Attribution (CC BY) licence to any Author Accepted Manuscript version arising.

\appendix

\renewcommand{\thesection}{\Alph{section}} 
\makeatletter
\renewcommand\@seccntformat[1]{\appendixname\ \csname the#1\endcsname.\hspace{0.5em}}
\makeatother

\section{Proofs and auxiliary results}
\label{app1}
\subsection{Proof of Theorem \ref{thm}}\label{proof_thm}
\begin{proof}
    First, we will show that
    \begin{equation}\label{eq:f}
        f(i,j)=\frac{1}{i+1}+i\sum_{j'=1}^{j-1}\left(\frac{B(i+j',j'+2)}{j'B(j',j'+1)}-\frac{B(i+j',j'+1)}{j'B(j',j')}\right)
    \end{equation}
    where $f(i,j):=\mathbb{P}\left(X>\max_{i'=1,\dots,i}U_{i'}\right)$ for $X\sim\text{Beta}(j,j)$ and $U_{i'}\sim\text{Uniform}(0,1)$ all independent. Indeed, $\max_{{i'}=1,\dots,i}U_{i'}$ has probability density function $\mathbbm{1}(u\in(0,1))iu^{i-1}$, so
    \begin{multline*}
        f(i,j)=\int_0^1ix^{i-1}(1-I_x(j,j))dx=1-\int_0^1ix^{i-1}I_x(j,j)dx \\ = 1-\int_0^1ix^{i-1}\left(I_x(j-1,j)-\frac{x^{j-1}(1-x)^{j}}{(j-1)B(j-1,j)}\right)dx \\ = 1-\int_0^1ix^{i-1}\left(I_x(j-1,j-1)+\frac{x^{j-1}(1-x)^{j-1}}{(j-1)B(j-1,j-1)}-\frac{x^{j-1}(1-x)^{j}}{(j-1)B(j-1,j)}\right)dx \\ = f(i,j-1) - \frac{iB(i+j-1,j)}{(j-1)B(j-1,j-1)}+\frac{iB(i+j-1,j+1)}{(j-1)B(j-1,j)},
    \end{multline*}
    where we have used the same recursive formulas for the regularised incomplete beta function as in the proof of Lemma \ref{lem:rec}. Moreover, by symmetry $f(i,1)=\frac{1}{i+1}$. Induction then gives us Equation \eqref{eq:f}. From this, we can see that lines 2--10 of Algorithm \ref{algo} compute $P^{(-n')}(S)=P(\bm{1}_{2k};S)$ in $\mathcal{O}(k2^k)$ operations, where $\bm{1}_{2k}$ is a vector of $2k$ ones.

    Further, supposing in the for loop from lines 12--33 we assumed that $P^{(n)}(S)=P(\bm{x}^n;S)$, then by Lemma \ref{lem:rec}
    \begin{align*}
        P^{(n+1)}(S)&=P(\bm{x}^n;S)+\begin{cases}
            (-1)^{s+1}\sum_{j'\notin S}\frac{b_{j'}(\bm{x}^n;S)P(\bm{x}^n;S\cup\{j'\})}{\sum_{j''\in S}x_{2j''+s}^n} &\text{if $j\in S$}\\ (-1)^s\frac{b_j(\bm{x}^n;S)P(\bm{x}^n;S\cup\{j\})}{x_{2j+s}^n} &\text{if $j\notin S$}
        \end{cases} \\ &= P(\bm{x}^{n+1};S).
    \end{align*}
    Thus, by induction, $P^{(n)}(\{j\})=P(\bm{x}^{(n)};\{j\})$ for all $j\in\{0,\dots,k-1\}$ and $n\in\{0,\dots,n\}$, as required. Moreover, since one single such update takes $\mathcal{O}(k)$ operations, and there are $\mathcal{O}(2^k)$ possible subsets $S$, one run of this for loop takes $\mathcal{O}(k2^k)$ operations. There are $n+n'$ such updates, so overall the algorithm takes $\mathcal{O}(k2^k(1+n+n'))=\mathcal{O}\left(k2^k\left(n-2k+\sum_{i=0}^{2k-1}x^0_i\right)\right)$ operations.
\end{proof}
\subsection{Proof of Lemma \ref{lem:rec}}\label{proof_lemma_diff}
\begin{proof}
    First, we define the regularised incomplete beta function $I_x$ as the cumulative distribution function of the Beta distribution, and recall that (by integrating by parts)
    \begin{equation*}\label{eq:lem1}
        I_x(a,b)=\frac{1}{aB(a,b)}x^a(1-x)^b+I_x(a+1,b).
    \end{equation*}
    As a result, for $s=0$ we get that for any $0\leq j,j'\leq k-1$ s.t. $j\neq j'$,
    \begin{multline*}
        P(\bm{x};\{j'\})=\int_0^1\frac{p^{x_{2j'}}(1-p)^{x_{2j'+1}}}{B(x_{2j'},x_{2j'+1})}\prod_{j''\neq j'}I_p(x_{2j''},x_{2j''+1})dp \\
        = \int_0^1\frac{p^{x_{2j}+x_{2j'}}(1-p)^{x_{2j+1}+x_{2j'+1}}}{x_{2j}B(x_{2j},x_{2j+1})B(x_{2j'},x_{2j'+1})}\prod_{j''\neq j,j'}I_p(x_{2j''},x_{2j''+1})dp + P(\bm{x}+\bm{e}_{2j};\{j'\}) \\
        = P(\bm{x}+\bm{e}_{2j};\{j'\}) + \frac{B(x_{2j}+x_{2j'},x_{2j+1}+x_{2j'+1})P(\bm{x};\{j,j'\})}{x_{2j}B(x_{2j},x_{2j+1})B(x_{2j'},x_{2j'+1})}
    \end{multline*}
    as required. Now, using this result,
    \begin{multline*}
        P(\bm{x}+\bm{e}_{2j};\{j\})-P(\bm{x};\{j\})=\left(1-\sum_{j''\neq j}P(\bm{x}+\bm{e}_{2j};\{j\})\right)-\left(1-\sum_{j''\neq j}P(\bm{x};\{j\})\right) \\
        = -\sum_{j''\neq j}\left(P(\bm{x}+\bm{e}_{2j};\{j\})-P(\bm{x};\{j\})\right) = -\sum_{j''\neq j}\frac{b_{j,j''}(\bm{x})}{x_{2j}}P(\bm{x};\{j,j''\}).
    \end{multline*}
    The case where $s=1$ is analogous.
\end{proof}

\subsection{Additional results}
In the paper, we made use of the following straightforward lemmas:
\begin{lemma}\label{lem:prior_decomp}
Let~$s_{i,j},n_{i,j}$ be the realisations of~$S_{i,j},N_{i,j}$. $\mathbb{Q}^i\equiv\mathbb{Q}^{i}_{0}\times\ldots\times\mathbb{Q}^{i}_{k-1}$ where $\mathbb{Q}^{i}_{j}$ for $0\leq j<k$ is defined by
\begin{equation}
    \mathbb{Q}^{i}_{j}(E):=\frac{\int_{B}x_{j}^{s_{i,j}}(1-x_{j})^{n_{i,j}-s_{i,j}}\mathbb{Q}_{j}(dx_{j})}{\int_0^1 x_{j}^{s_{i,j}}(1-x_{j})^{n_{i,j}-s_{i,j}}\mathbb{Q}_{j}(dx_{j})}
\end{equation}
for some $\mathbb{Q}_{j}$-measurable set $E\subseteq[0,1]$. As a result, it follows that
\begin{equation}
    \pi^{\mathbb{Q}}_{\text{S-BRAR},j}({\bm H}_i)=\frac{\int_{x_j>\max_{j'\neq j}x_{j'}}\prod_{j'=0}^{k-1}x_{j'}^{s_{i,j'}}(1-x_{j'})^{n_{i,j'}-s_{i,j'}}\mathbb{Q}_{j'}(dx_{j'})}{\int_{[0,1]^k}\prod_{j'=0}^{k-1}x_{j'}^{s_{i,j'}}(1-x_{j'})^{n_{i,j'}-s_{i,j'}}\mathbb{Q}_{j'}(dx_{j'})}.
\end{equation}
\end{lemma}
\begin{proof}
    Let $f_{p_j}$ be the (prior) probability density function for $\mathbb{Q}_j$, $f_{{\bm H}_i|{\bm p}}$ be that of the distribution of ${\bm H}_i|{\bm p}$, and $f_{{\bm p}|{\bm H}_i}$ be that of $\mathbb{Q}^i$. Then,
    \begin{equation*}
         f_{{\bm H}_i|{\bm p}}(\bm{h}_i|\bm{p})=\prod_{j=0}^{k-1} p_j^{s_{i,j}}(1-p_j)^{n_{i,j}-s_{i,j}}\prod_{i'=0}^{i-1}\pi(\bm{h}_{i'})^{A_{i'+1}}(1-\pi(\bm{h}_{i'}))^{1-A_{i'+1}},
    \end{equation*}
    so that by Bayes' theorem,
    \begin{align*}
        f_{{\bm p}|{\bm H}_i}(\bm{p}|\bm{h}_i)=\frac{\prod_{j=0}^{k-1}f_{p_j}(p_j)p_{j}^{s_{i,j}}(1-p_{j})^{n_{i,j}-s_{i,j}}}{\int_{[0,1]^k}\prod_{j=0}^{k-1}f_{p_j}(p_j)p_{j}^{s_{i,j}}(1-p_{j})^{n_{i,j}-s_{i,j}}d\bm{p}}
        \\ =\prod_{j=0}^{k-1}\frac{f_{p_j}(p_j)p_{j}^{s_{i,j}}(1-p_{j})^{n_{i,j}-s_{i,j}}}{\int_0^1 f_{p_j}(p_j)p_{j}^{s_{i,j}}(1-p_{j})^{n_{i,j}-s_{i,j}}dp_j}.
    \end{align*}
\end{proof}

\begin{lemma}\label{lem:conj}
    Let~$s_{i,j},n_{i,j}$ be the realisations of~$S_{i,j},N_{i,j}$. If for all $0\leq j<k$, \\$\mathbb{Q}_j\overset{d}{\equiv}\text{Beta}(a_j,b_j)$ for some $\bm{a},\bm{b}\in[1,\infty)^k$, then $\mathbb{Q}^i\overset{d}{\equiv}\bigtimes_{j=0}^{k-1}\text{Beta}(a_j+s_{i,j},b_j+n_{i,j}-s_{i,j})$.
\end{lemma}
\begin{proof}
    From the first part of Lemma \ref{lem:prior_decomp}, for all Lebesgue measurable sets~$E\subseteq[0,1]$,
    \begin{align*}
        \mathbb{Q}^{i}_{j}(E) & \propto \int_{B}x_{j}^{s_{i,j}}(1-x_{j})^{n_{i,j}-s_{i,j}}\frac{x_j^{a_j-1}(1-x_j)^{b_j-1}}{B(a_j,b_j)}dx_j \\
        & \propto \int_{B}x_{j}^{a_j+s_{i,j}-1}(1-x_{j})^{b_j+n_{i,j}-s_{i,j}-1}dx_j,
    \end{align*}
    which is of the required form to give the required result again by Lemma \ref{lem:prior_decomp}.
\end{proof}
\subsubsection{Calculating critical values}\label{appc}
Moreover, regarding the UX($\mathcal{P},\alpha$) and PP($\mathcal{P},p,\alpha$) tests, note that
\begin{equation}
    c(\mathcal{P},\alpha):=\inf\{c\in[0,1]:\:\sup_{p\in[0,1]}\mathbb{P}^{\pi}(\max_{j\in\{0,\ldots,k-1\}}T_j({\bm H}_{\bar{T}_c})>c\mid p_0=\ldots=p_{k-1}=p)\leq\alpha\}
\end{equation}
and
\begin{equation}
    c(\mathcal{P},p,\alpha):=\inf\{c\in[0,1]:\:\mathbb{P}^{\pi}(\max_{j\in\{0,\ldots,k-1\}}T_j({\bm H}_{\bar{T}_c})>c\mid p_0=\ldots=p_{k-1}=p)\leq\alpha\},
\end{equation}
where $\bar{T}_c$ is the last patient to be randomised before the trial is stopped as a consequence of using the rejection rule $\mathcal{R}(\mathcal{P},c)$. 

The latter was calculated using all the possible randomisation probabilities to calculate probability of every possible final state in an efficient way known as using ``forward equations", where different trial realisations are said to have the same final state iff the values of $T_j$ are equal for every possible $j$ at the end of the trial. Conveniently, a final state is determined by the number of positive and negative responses on each arm, since this is all that $\mathbb{Q}^i$ for any $i$, and so $T_j$, depends on. The probability of each final state is matched with the maximal value of $T_j$ over all $j$ at that state, and this larger set is then sorted based on the value of the maximal test statistic so that the probabilities of the largest ones can be summed until $\alpha$ is reached. 

\section{Code availability}

The code corresponding to this paper can be found at: \url{https://github.com/dkaddaj/ExactTS}.





\end{document}